\documentclass[a4j,dvipdfmx,uplatex]{article}
\usepackage{amsmath, amssymb}
\usepackage{amsthm}
\usepackage{geometry}
\usepackage{algorithm}
\usepackage{algpseudocode}
\usepackage{here}
\usepackage{tikz}
\usepackage{comment}
\usetikzlibrary{cd}

\theoremstyle{plain}
\newtheorem{thm}{Theorem}[section]
\newtheorem{lem}[thm]{Lemma}
\newtheorem{prop}[thm]{Proposition}
\newtheorem{cor}[thm]{Corollary}

\theoremstyle{definition}

\newtheorem{dfn}[thm]{Definition}

\theoremstyle{remark}
\newtheorem{rem}[thm]{Remark}

\theoremstyle{plain}
\newtheorem*{thm*}{Theorem}
\newtheorem*{lem*}{Lemma}
\newtheorem*{prop*}{Proposition}
\newtheorem*{cor*}{Corollary}
\newtheorem*{conj*}{Conjecture}

\theoremstyle{definition}
\newtheorem*{ass*}{Assumption}
\newtheorem*{dfn*}{Definition}

\theoremstyle{remark}
\newtheorem*{rem*}{Remark}

\DeclareMathOperator{\ord}{ord}
\DeclareMathOperator{\lcm}{lcm}

\makeatletter
\newcommand{\ie}{\emph{i.e.}\@ifnextchar.{\!\@gobble}{}}
\newcommand{\eg}{\emph{e.g.}\@ifnextchar.{\!\@gobble}{}}
\newcommand{\etc}{etc\@ifnextchar.{}{.\@}}
\makeatother
\title{Cooley-Tukey FFT over $\mathbb{Q}_p$ via Unramified Cyclotomic Extension}
\author{Hiromasa Kondo}
\date{}

\begin{document}
\maketitle
\begin{abstract}
The reason why Cooley-Tukey Fast Fourier Transform (FFT) over $\mathbb{Q}$
can be efficiently implemented using complex roots of unity  
is that the cyclotomic extensions of the completion $\mathbb{R}$ of $\mathbb{Q}$ are at most quadratic,  
and that roots of unity in $\mathbb{C}$ can be evaluated quickly.  
In this paper, we investigate a $p$-adic analogue of this efficient FFT. 
A naive application of this idea—such as invoking well-known algorithms like the Cantor-Zassenhaus algorithm  
or Hensel's lemma for polynomials to compute roots of unity—would incur a cost quadratic in the degree of the input polynomial.  
This would eliminate the computational advantage of using FFT in the first place. 
We present a method for computing roots of unity with lower complexity than the FFT computation itself.  
This suggests the possibility of designing new FFT algorithms for rational numbers.
As a simple application, 
we construct an $O(N^{1+o(1)})$-time FFT algorithm over $\mathbb{Q}_p$ for fixed $p$.
\end{abstract}

\section{Introduction}
Throughout this paper, for any positive integer $s$, we denote by $\zeta_s$ a fixed primitive $s$-th root of unity.  
We also denote by $\Phi_s$ the $s$-th cyclotomic polynomial.  
When we wish to make the indeterminate explicit, we write $\Phi_s(X)$, and so on.  
The ring of $p$-adic integers is denoted by $\mathbb{Z}_p$,  
and the finite field with $q$ elements is denoted by $\mathbb{F}_q$.  
The field of fractions of $\mathbb{Z}_p$ is denoted by $\mathbb{Q}_p$.

The well-known Cooley-Tukey FFT algorithm~\cite{CooleyTukey,ModernComputerAlgebra}  
is an efficient method for computing the discrete Fourier transform (DFT) of a polynomial $f$  
at the points $1, \zeta_M, \zeta_M^2, \ldots, \zeta_M^{M-1}$, where $M$ is a highly composite positive integer.  
If $M$ has the prime factorization $M = p_1^{v_1} p_2^{v_2} \cdots p_n^{v_n}$,  
then the computation requires $O(M(p_1v_1 + p_2v_2 + \cdots + p_nv_n))$ additions and multiplications.  
This technique enables efficient computation of convolution-type operations such as polynomial multiplication,  
and has had significant impact in both science and engineering. 
Moreover, using Bluestein's trick~\cite{Bluestein,PreparataSawate},  
a DFT of arbitrary length $M$ can be reduced to a convolution computation,  
which in many fields—including the complex numbers—can be efficiently reduced to the case where $M$ is highly composite.

In order to perform a Cooley-Tukey FFT, it is necessary that a primitive $M$-th root of unity $\zeta_M$ exist in the base field.  
Alternatively, as suggested in~\cite{Pollard}, one can adjoin such a root if it does not exist.  
However, as pointed out in~\cite{Cantor_Finite_Fields} in the finite field setting,  
if one simply adjoins $\zeta_M$—for example, by choosing $M$ to be a power of $2$—then the degree of the resulting extension becomes comparable to $M$,  
and the total cost of the FFT procedure becomes $\Omega(M^3)$.  

The same issue arises over the field of rational numbers $\mathbb{Q}$.  
Since every cyclotomic polynomial is irreducible over $\mathbb{Q}$,  
the extension $\mathbb{Q}(\zeta_M)/\mathbb{Q}$ has degree $\varphi(M)$,  
so the computational cost of the FFT procedure becomes similarly large.  
From this perspective, one reason FFT over the complex numbers works so efficiently  
is that the cyclotomic extension $\mathbb{R}(\zeta_M)$ of the real numbers $\mathbb{R}$—the completion of $\mathbb{Q}$—has degree at most $2$ over $\mathbb{R}$.

In this paper, we study the $p$-adic analogue of this idea:  
instead of completing $\mathbb{Q}$ to $\mathbb{R}$, we complete it to $\mathbb{Q}_p$,  
and investigate Cooley-Tukey FFTs over $p$-adic cyclotomic extensions $\mathbb{Q}_p(\zeta_M)$  
with small extension degree $d = [\mathbb{Q}_p(\zeta_M): \mathbb{Q}_p]$.  
The most nontrivial aspect of this approach is that $M$-th roots of unity must be efficiently computable,  
which will be discussed in detail in Section~3.


We evaluate the computational complexity in terms of the number of arithmetic operations over $\mathbb{Q}_p$ and its residue field $\mathbb{F}_p$.  
Under this convention, we assume that multiplication in a degree-$d$ extension of $\mathbb{Q}_p$ is performed without using FFT,  
and therefore each multiplication costs $O(d^2)$ operations. 
This paper does not attempt to optimize this part of the computation, as its treatment is largely independent of the main algorithm we aim to present. 

\begin{center}
\includegraphics[width=8cm]{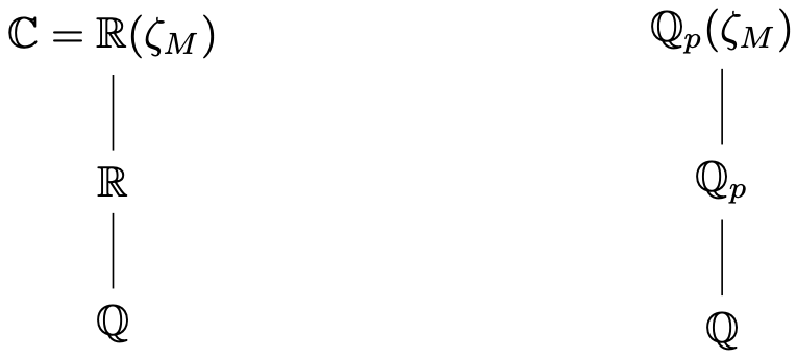}
\end{center}

\subsection{Relation to existing algebraic FFTs}
We describe the relationship between the FFT considered in this paper and existing algebraic variants of the FFT, including the number-theoretic transform, additive FFT, and cyclotomic FFT.
\begin{itemize}

\item
The \emph{number-theoretic transform} (NTT) refers to an FFT over a residue ring of $\mathbb{Z}$.  
The most basic version of NTT performs reduction modulo a prime number~\cite{Pollard}.  
Since the multiplicative group of the field $\mathbb{Z}/p\mathbb{Z}$ is isomorphic to $\mathbb{Z}/(p-1)\mathbb{Z}$,  
if $p - 1$ is divisible by a highly composite integer $M$,  
then a primitive $M$-th root of unity exists in this field and FFT can be performed.  
Our algorithm can be viewed as a $p$-adic refinement of this idea.  
When $M \mid (p - 1)$, i.e., when a primitive $M$-th root of unity $\zeta_M$ exists in $\mathbb{Z}_p$,  
it can be computed efficiently using the simple-root Hensel's lemma (Proposition~\ref{HLiftFast})~\cite{IntegerMultiplication}.  
However, it is unclear how efficiently $\zeta_M$ can be computed when working over low-degree algebraic extensions of $\mathbb{Z}_p$.  
This question is one of the central themes of this paper.

\item
The \emph{additive FFT} is a method for FFT over fields of positive characteristic.  
The term ``additive'' comes from the fact that, while ordinary FFT evaluates at the multiplicative subgroup  
$\{1, \zeta_M, \ldots, \zeta_M^{M-1}\} \subset \mathbb{F}$,  
the additive FFT evaluates over arbitrary finite additive subgroups of $\mathbb{F}$.  
For example, if $\mathbb{F}$ is of characteristic $2$ and contains $\mathbb{F}_{2^{n}} = \{x\in\mathbb{F}\mid x^{2^{n}}-x=0\}$,  
then $\mathbb{F}_{2^{n}}$ forms a finite additive subgroup of $\mathbb{F}$,  
and the values of a polynomial at its elements can be efficiently computed.  
In contrast to the additive FFT, the ordinary FFT is sometimes referred to as the multiplicative FFT.  

Historically, the idea of additive FFT was introduced by Wang and Zhu~\cite{WangZhu} and Cantor~\cite{Cantor_Finite_Fields},  
and has since undergone various improvements.  
For a concise and accessible overview of the algorithm, the paper by Gao and Mateer~\cite{GaoMateer} is highly recommended.  
The history of additive FFT is detailed in Coxon's recent work~\cite{Coxon}.

Our FFT procedure consists of two parts: one over finite fields and the other over $p$-adic fields.  
The finite field part corresponds to a multiplicative FFT.  
On the other hand, additive FFTs crucially rely on the characteristic being positive,  
and it is unclear whether such techniques can be adapted in any meaningful way to the $p$-adic setting.  

\item
The \emph{cyclotomic FFT}~\cite{TrifonovFedorenko} is another FFT method over finite fields.  
Although this method does not offer the best asymptotic complexity,  
it can efficiently compute DFTs of specific short lengths.  
As this technique also seems to rely on the characteristic being positive,  
its connection to our algorithm is uncertain. 

\item 
We also discuss \emph{multipoint evaluation and interpolation}, which generalize the relationship between the FFT and inverse FFT commonly used in polynomial multiplication (see, for example, \cite{MultipointEvaluationInterpolation}).  
In \emph{multipoint evaluation}, the task is to compute the values $f(\alpha_1), \ldots, f(\alpha_M)$ of a given polynomial $f$ at $M$ distinct points $\alpha_1, \ldots, \alpha_M$ (called a multipoint).  
Conversely, in \emph{interpolation}, one is given the values $f(\alpha_1), \ldots, f(\alpha_M)$ and aims to reconstruct the polynomial $f$.

Regarding multipoint evaluation, Kedlaya and Umans developed a theoretically important algorithm that performs asymptotically fast evaluation at arbitrary multipoints over a wide class of fields~\cite{2008, Kedlaya-Umans}.  
While this result is of great theoretical significance in computer algebra, it is known that the algorithm is not efficient for problems of practical size on current hardware~\cite{Revisited}.  
For this reason, we do not rely on their algorithm in this paper.

\end{itemize}

\section{Cyclotomic extensions of $\mathbb{Q}_p$: degrees and valuation structures}
This chapter does not contain any new mathematical results of our own.
Since it seems natural to also investigate the valuation structure of cyclotomic extensions of $\mathbb{Q}_p$, in addition to their degrees of extension, we will describe that as well.
The discussion on the valuation structure is based on Serre \cite{Serre}.
We are not aware of any other suitable references beyond this book,
and since the relevant parts in Serre's book are scattered across a wide range,
making it rather cumbersome to piece them together,
we believe that presenting a summarized account here is worthwhile.

The following lemma is clear from considering how the Frobenius automorphism acts on the set of roots of $\Phi_s$. 
\begin{lem}\label{Fp}
Let $p$ be a prime number and $s$ be a positive integer relatively prime to $p$. 
$$[ \mathbb{F}_p (\zeta_s ) : \mathbb{F}_p] = \ord_s p$$
holds, where $\ord_s p$ is the smallest positive integer $r$ such that $p^r \equiv 1 \pmod s$. 
\end{lem}

\begin{dfn}
A \emph{discrete valuation ring (DVR)} is a commutative ring $A$ satisfying the following equivalent conditions:
\begin{itemize}
\item $A$ is a PID, local, and not a field;
\item $A$ is Noetherian, local, and its maximal ideal is generated by a non-nilpotent element.
\end{itemize}
Moreover, if $\pi \in A$ is a generator of the maximal ideal of the discrete valuation ring $A$, then any nonzero element $y \in A \setminus \{0\}$ can be uniquely written in the form $\pi^n u$ for some $n \in \mathbb{Z}_{\geq 0}$ and $u \in A^\times$. This $n$ is called the \emph{valuation} of $y$ and denoted by $v(y)$. For elements $x, y \in A$, we define
\[
|x - y| =
\begin{cases}
\exp(-v(x - y)) & \text{if } x - y \neq 0, \\
0 & \text{if } x - y = 0,
\end{cases}
\]
giving $A$ the structure of a metric space. If $A$ is complete with respect to this metric, we call $A$ a \emph{complete discrete valuation ring}.
\end{dfn}

$\mathbb{Z}_p$ is an example of a complete discrete valuation ring.

\begin{lem}[Hensel lifting]\label{HLift}
Let $A$ be a DVR, and let $\mathfrak{m}$ denote the maximal ideal of $A$.  
Suppose we are given a monic polynomial $h \in A[X]$, monic polynomials $f, g \in A[X]$ of degree at least $1$, polynomials $a, b \in A[X]$, and an integer $k \in \mathbb{Z}_{>0}$ satisfying the following conditions:
\begin{itemize}
\item $h \equiv fg \pmod{\mathfrak{m}^k}$;
\item $af + bg \equiv 1 \pmod{\mathfrak{m}}$.
\end{itemize}

Then there exist $\delta_f, \delta_g \in \mathfrak{m}^kA[X]$, unique modulo $\mathfrak{m}^{k+1}$, satisfying:
\begin{itemize}
\item $\deg \delta_f < \deg f$ and $\deg \delta_g < \deg g$;
\item $h \equiv (f+\delta_f)(g+\delta_g) \pmod{\mathfrak{m}^{k+1}}$;
\item $a(f+\delta_f) + b(g+\delta_g) \equiv 1 \pmod{\mathfrak{m}}$.
\end{itemize}
\end{lem}

\begin{proof}
We first prove the existence.  
Let $q_g$ and $\delta_g$ be the quotient and remainder, respectively, when dividing $a(h-fg)$ by $g$.  
Similarly, let $q_f$ and $\delta_f$ be the quotient and remainder when dividing $a(h-fg)$ by $f$.  
It is clear that $\deg \delta_f < \deg f$ and $\deg \delta_g < \deg g$.  
Since $h - fg \in \mathfrak{m}^k$, we have $\delta_g, \delta_f \in \mathfrak{m}^k$.  
Then,
\begin{align*}
h - (f + \delta_f)(g + \delta_g) 
&\equiv (h - fg) - (f\delta_g + g\delta_f) \pmod{\mathfrak{m}^{k+1}} \\
&\equiv -(af + bg - 1)(h - fg) + fg(q_f + q_g) \\
&\equiv fg(q_f + q_g).
\end{align*}
In the expression $(h-fg)-(f\delta_g+g\delta_f)\equiv fg(q_f+q_g)\pmod{\mathfrak{m}^{k+1}}$,  
the left-hand side has degree at most $\deg f + \deg g - 1$, while $fg$ is a monic polynomial of degree $\deg f + \deg g$.  
Thus, $q_f + q_g$ must belong to $\mathfrak{m}^{k+1}$.  
Therefore, $h - (f+\delta_f)(g+\delta_g) \in \mathfrak{m}^{k+1}$.

We next prove uniqueness.  
Suppose that $\delta_g'$ and $\delta_f'$ also satisfy the same conditions as $\delta_g$ and $\delta_f$.  
Then, since
\[
(f + \delta_f)(g + \delta_g) \equiv (f + \delta_f')(g + \delta_g') \pmod{\mathfrak{m}^{k+1}},
\]
we have
\begin{equation*}\label{F}
f(\delta_g - \delta_g') + g(\delta_f - \delta_f') \equiv 0 \pmod{\mathfrak{m}^{k+1}}.
\end{equation*}
From this,
\begin{align*}
\delta_g - \delta_g' 
&\equiv (af + bg)(\delta_g - \delta_g') \pmod{\mathfrak{m}^{k+1}} \\
&\equiv g\left(b(\delta_g - \delta_g') - a(\delta_f - \delta_f')\right).
\end{align*}
Since $\deg(\delta_g - \delta_g') < \deg g$ and $g$ is monic,  
we must have $b(\delta_g - \delta_g') - a(\delta_f - \delta_f') \equiv 0 \pmod{\mathfrak{m}^{k+1}}$.  
Thus, $\delta_g - \delta_g' \in \mathfrak{m}^{k+1}$, and similarly, $\delta_f - \delta_f' \in \mathfrak{m}^{k+1}$.
\end{proof}

\begin{cor}[Hensel's lemma]\label{HLGeneral}
Let $A$ be a complete DVR, and let $\mathfrak{m}$ denote the maximal ideal of $A$.  
Let $h \in A[X]$ be a monic polynomial.  
We denote the image of $A[X]$ under the natural projection $A[X] \to (A/\mathfrak{m})[X]$ by adding a bar, e.g., $\bar{h}$.  
Suppose there exist coprime monic polynomials $f, g \in (A/\mathfrak{m})[X]$ of degree at least $1$ such that $\bar{h} = fg$.  
Then there uniquely exist polynomials $(\tilde{f}, \tilde{g}) \in A[X]^2$ satisfying the following conditions:
\begin{itemize}
\item $\tilde{f}$ and $\tilde{g}$ are monic, with $\deg \tilde{f} = \deg f$ and $\deg \tilde{g} = \deg g$;
\item $h = \tilde{f} \tilde{g}$;
\item $\bar{\tilde{f}} = f$ and $\bar{\tilde{g}} = g$.
\end{itemize}
\end{cor}

\begin{prop}
Let $p$ be a prime number and $s$ a positive integer relatively prime to $p$. Then we have
\[
[\mathbb{Q}_p(\zeta_s):\mathbb{Q}_p] = \ord_s p.
\]
\end{prop}
\begin{proof}
Let $f$ be a monic irreducible factor of $\Phi_s \in \mathbb{Q}_p[X]$.  
All the coefficients of $f$ lie in $\mathbb{Z}_p$.  
Let $\bar{f}$ and $\bar{\Phi_s}$ denote the images of $f$ and $\Phi_s$ under the natural projection $\mathbb{Z}_p[X] \to \mathbb{F}_p[X]$.  
By Lemma~\ref{Fp}, we have $\deg f \geq \ord_s p$, and hence $[\mathbb{Q}_p(\zeta_s):\mathbb{Q}_p] \geq \ord_s p$.  

Conversely, by Hensel's lemma, any irreducible factor of $\bar{\Phi_s}$ lifts to a factor of $\Phi_s$.  
Therefore, $[\mathbb{Q}_p(\zeta_s):\mathbb{Q}_p] \leq \ord_s p$.
\end{proof}

\begin{prop}\label{Unramified}
Let $A$ be a DVR, $\mathfrak{m}$ its maximal ideal, and $K$ the field of fractions of $A$.  
Suppose that $f \in A[X]$ is a polynomial such that its image $\bar{f} \in (A/\mathfrak{m})[X]$ is irreducible.  
Set $B_f = A[X]/(f)$.  
Then $B_f$ is a DVR whose maximal ideal is $\mathfrak{m}B_f$.  
Moreover, the field of fractions of $B_f$ is $K[X]/(f)$.
\end{prop}

\begin{proof}
We have $B_f/\mathfrak{m}B_f \cong (A/\mathfrak{m})[X]/(\bar{f})$,  
which is a field since $\bar{f}$ is irreducible.  
Thus, $\mathfrak{m}B_f$ is the maximal ideal of $B_f$.
Suppose $\mathfrak{n}$ is a maximal ideal of $B_f$.  
If $\mathfrak{n} \not\supset \mathfrak{m}B_f$, then $\mathfrak{n} + \mathfrak{m}B_f = B_f$.  
Since $B_f$ is a finitely generated $A$-module, this contradicts Nakayama's lemma.  
Hence, $\mathfrak{n} = \mathfrak{m}B_f$, and $B_f$ is a local ring.
Since $\mathfrak{m}$ is non-nilpotent in $A$, it remains non-nilpotent in $B_f$.  
Thus, $B_f$ is Noetherian, local, and has a non-nilpotent maximal ideal; therefore, $B_f$ is a DVR.
Finally, it is clear that the field of fractions of $B_f$ is $K[X]/(f)$.
\end{proof}

Let $f \in \mathbb{Z}_p[X]$ be an irreducible factor of $\Phi_s \in \mathbb{Z}_p[X]$. 
By proposition \ref{Unramified}, $\mathbb{Z}_p[\zeta_s]:=\mathbb{Z}_p[X]/f$ is a DVR with maximal ideal $p\mathbb{Z}_p[\zeta_s]$. 

\begin{center}
\includegraphics[width=8cm]{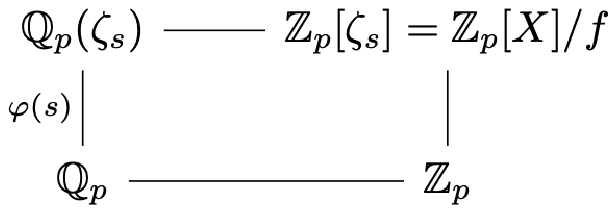}
\end{center}

We next describe fields of the form $\mathbb{Q}_p(\zeta_{sp^n})$ for $n \geq 1$.
\begin{prop}\label{TotallyRamified}
Let $A$ be a DVR, $\mathfrak{m}$ its maximal ideal, and $K$ the field of fractions of $A$.  
Suppose that $f \in A[X]$ is of Eisenstein type, that is, if we write
\[
f = X^n + a_{n-1}X^{n-1} + a_{n-2}X^{n-2} + \cdots + a_0,
\]
then $a_i \in \mathfrak{m}$ for all $i$ and $a_0 \notin \mathfrak{m}^2$.  
(In particular, $f$ is irreducible over $K$.)

Set $B_f = A[X]/(f)$.  
Then $B_f$ is a DVR whose maximal ideal is $XB_f$.  
Moreover, the field of fractions of $B_f$ is $K[X]/(f)$.
\end{prop}

\begin{proof}
It is clear that the field of fractions of $B_f$ is $K[X]/(f)$.
We have $B_f/XB_f \cong A/\mathfrak{m}$, which is a field.  
Thus, $XB_f$ is a maximal ideal of $B_f$.
As in Proposition~\ref{Unramified}, it follows that $B_f$ is a local ring.
Since $a_0$ is a uniformizer of $A$, the element $a_0 \in B_f$ is non-nilpotent.  
Moreover, in $B_f$ we have the relation
\[
a_0 = (-X^{n-1} - a_{n-1}X^{n-2} - \cdots - a_1)X,
\]
which implies that $X \in B_f$ is also non-nilpotent.
Therefore, $B_f$ is a DVR whose maximal ideal is $XB_f$.
\end{proof}

\begin{rem}
In both Proposition~\ref{Unramified} and Proposition~\ref{TotallyRamified},  
the ring $B_f$ is the integral closure of $A$ in $K[X]/f$.  
The extension in Proposition~\ref{Unramified} is unramified,  
while the extension in Proposition~\ref{TotallyRamified} is totally ramified.  
(See \cite{Serre}, Chapter~1, Section~6.)
\end{rem}

The tower of field extensions
\[
\cdots / \mathbb{Q}_p(\zeta_{sp^2}) / \mathbb{Q}_p(\zeta_{sp}) / \mathbb{Q}_p(\zeta_s)
\]
is obtained by repeatedly applying Proposition~\ref{TotallyRamified}.  

Over $\mathbb{Z}_p[\zeta_s]$, the polynomial
\[
\Phi_p(X+1) = X^{p-1} + pX^{p-2} + \cdots + p
\]
is of Eisenstein type, and hence $[\mathbb{Q}_p(\zeta_{sp}):\mathbb{Q}_p(\zeta_s)] = p-1$.  
Moreover, by Proposition~\ref{TotallyRamified}, the ring
\[
\mathbb{Z}_p[\zeta_{sp}] := \mathbb{Z}_p(\zeta_s)[X]/\Phi_p
\]
is a DVR whose maximal ideal is $(\zeta_p - 1)\mathbb{Z}_p[\zeta_{sp}]$.


\begin{center}
\includegraphics[width=8cm]{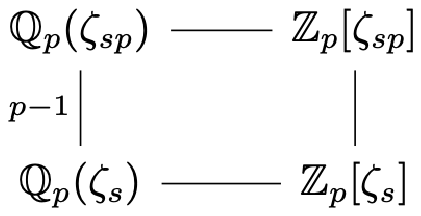}
\end{center}

Fix $n \geq 1$.  
Suppose that $\mathbb{Z}_p[\zeta_{sp^n}]$ is a DVR whose maximal ideal is $(\zeta_{p^n} - 1)\mathbb{Z}_p[\zeta_{sp^n}]$.  
Over $\mathbb{Q}_p(\zeta_{sp^n})$, the polynomial $X^p - \zeta_{p^n}$, after replacing $X$ by $X+1$, becomes
\[
X^p + pX^{p-1} + \cdots + (1 - \zeta_{p^n}),
\]
which is of Eisenstein type.  
Thus, $[\mathbb{Q}_p(\zeta_{sp^{n+1}}) : \mathbb{Q}_p(\zeta_{sp^n})] = p$,  
and by Proposition~\ref{TotallyRamified}, the ring $\mathbb{Z}_p[\zeta_{sp^{n+1}}]:=\mathbb{Z}_p[\zeta_{sp^n}][X]/(X^p - \zeta_{p^n})$  
is a DVR whose maximal ideal is $(\zeta_{p^{n+1}} - 1)\mathbb{Z}_p[\zeta_{sp^{n+1}}]$.  

Therefore, by induction, we conclude that for any $n \geq 1$,  
$[\mathbb{Q}_p(\zeta_{sp^{n+1}}) : \mathbb{Q}_p(\zeta_{sp^n})] = p$.


\begin{center}
\includegraphics[width=8cm]{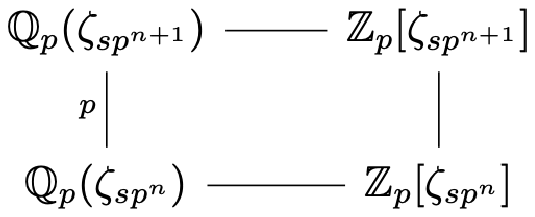}
\end{center}

Thus, the degree of cyclotomic extensions over $\mathbb{Q}_p$ is determined as follows.
\begin{thm}
Let $p$ be a prime number, $n$ a nonnegative integer, and $s$ a positive integer relatively prime to $p$. Then we have
\[
[\mathbb{Q}_p(\zeta_{sp^n}) : \mathbb{Q}_p] = (\ord_s p)\varphi(p^n).
\]
\end{thm}

In particular, in order to obtain a root of unity $\zeta_{sp^n}$ with a smaller extension degree,  
it is preferable to choose $n$ as small as possible.  
Moreover, we obtain the following lower bound for the extension degree in terms of the size of $s$:
\[
[\mathbb{Q}_p(\zeta_s) : \mathbb{Q}_p] \geq \log_p (s+1).
\]

\section{Discussion on computation methods and complexity}
In the previous section, we computed the degree of cyclotomic extensions of $\mathbb{Q}_p$.  
However, in order to actually perform Cooley-Tukey FFT, what matters is whether a corresponding root of unity can be explicitly constructed.  
By ``computing'' a root of unity $\zeta$, we mean the following:
\begin{itemize}
  \item an irreducible polynomial $F \in \mathbb{Q}_p[X]$ of the same degree as the minimal polynomial of $\zeta$, such that $\zeta$ belongs to $\mathbb{Q}_p[X]/F$, and
  \item an explicit representation of $\zeta$ in $\mathbb{Q}_p[X]/F$.
\end{itemize}

For example, if $F$ is the minimal polynomial of $\zeta$, then we may represent $\zeta$ by the indeterminate $X$ in $\mathbb{Q}_p[X]/F$.  
In the case of $\mathbb{C}$, the usual expression of elements as $(\text{real part}) + (\text{imaginary part}) \sqrt{-1}$  
corresponds to the choice $F = X^2 + 1$.  
Under this representation, a primitive $3$rd root of unity in $\mathbb{C}$ is expressed as  
$-0.5 + (0.866\ldots)\sqrt{-1}$. 

In this section, we assume the following conditions, which are suitable for performing FFT and are described in an axiomatic manner:
\begin{itemize}
  \item $p \nmid s$,
  \item $s$ is factored as $s = p_1^{v_1} p_2^{v_2} \cdots p_n^{v_n}$ into small prime factors in a form that can be efficiently computed, and
  \item the extension degree $d := [\mathbb{Q}_p(\zeta_s) : \mathbb{Q}_p]$ is small compared to $s$ (say, $d < \sqrt{ \frac{s}{v_1 p_1 + v_2 p_2 + \cdots + v_n p_n} }$).
\end{itemize}

Under these assumptions, we show that the minimal polynomial of $\zeta_s$ over $\mathbb{Q}_p$ can be computed efficiently.  
Specifically, since the FFT algorithm itself involves
$O\bigl(d^2 s (v_1 p_1 + v_2 p_2 + \cdots + v_n p_n)\bigr)$
additions, subtractions, and multiplications over $\mathbb{Q}_p$,  
we will show that the cost of computing the minimal polynomial is strictly smaller than this.
As in the computational content of the previous section,  
the method is to first find the minimal polynomial over $\mathbb{F}_p$ and then lift it to characteristic zero.

There are several possible settings in which this problem can be considered:
one can vary both $p$ and $s$, or vary $s$ while keeping $p$ fixed, or vice versa.  
In any of these settings, determining the most suitable combination of $p$ and $s$ is likely to be a difficult problem.  
In the next section, we present an example of the case where $p$ is fixed and $s$ varies.

\subsection{Computing a root of unity in $\mathbb{F}_p$}
We begin by computing a minimal polynomial of $\zeta_s$ over $\mathbb{F}_p$.
A naive approach to this problem is to first compute the cyclotomic polynomial $\Phi_s$,  
and then apply the equal-degree factorization step from a polynomial factorization algorithm over finite fields.

There is a vast literature on polynomial factorization over finite fields,  
and among them, the most practically significant milestone is the Cantor-Zassenhaus algorithm~\cite{CZ}.  
However, applying this algorithm here would incur a runtime of $\Omega(\deg^2 \Phi_s)$,  
making it unsuitable for our purpose of fast Fourier transform computation.
The fastest known algorithm for equal-degree decomposition is due to Kedlaya and Umans~\cite{Kedlaya-Umans},  
which improves upon earlier work by von zur Gathen and Shoup~\cite{GathenShoup}.  
Its runtime is nearly linear in $\deg \Phi_s$,  
but it is known to be inefficient in practice for problem sizes within the scope of current computing capabilities.

Therefore, in this paper we propose a randomized algorithm that computes the minimal polynomial of $\zeta_s$ in a more suitable way.
The outline of the algorithm is as follows.  
We start with the base field $\mathbb{F}_p$ and gradually extend it step by step as follows:


\begin{center}
\includegraphics[width=8cm]{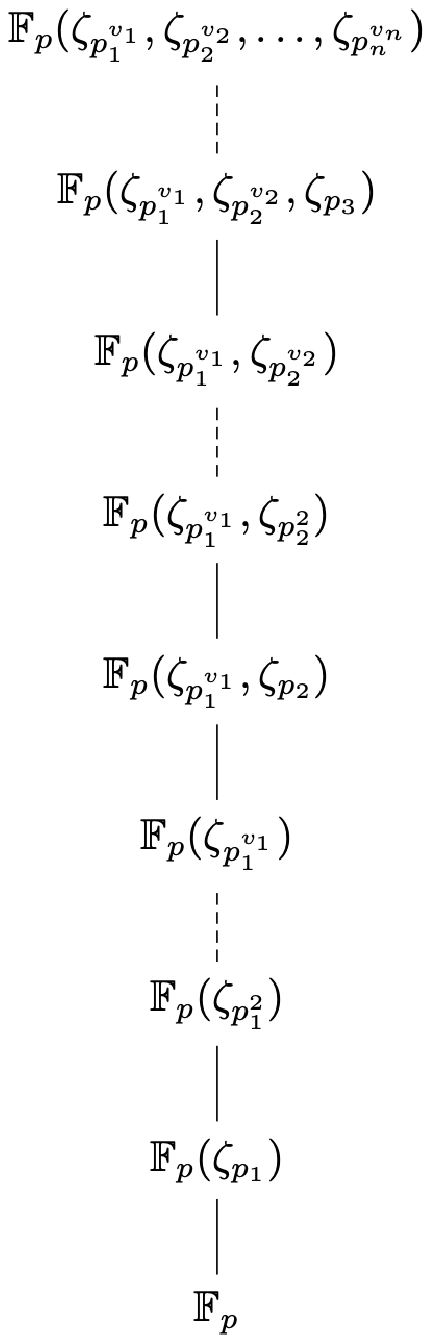}
\end{center}

Since the total extension has degree exactly $d$,  
each intermediate extension can be computed efficiently.
Each such extension is computed using the Cantor-Zassenhaus algorithm.  
That is, when adjoining $\zeta_{p_i}$, we factor the cyclotomic polynomial $\Phi_{p_i}$ to obtain $\zeta_{p_i}$.  
When adjoining $\zeta_{{p_i}^v}$ for $v > 1$, we factor the polynomial
$X^{p_i} - \zeta_{{p_i}^{v-1}}$
to obtain $\zeta_{{p_i}^v}$.

Each extension can be computed in
$O(d^2 p_i^2 \log p_i)$
operations over $\mathbb{F}_p$ (additions, subtractions, and multiplications).  
Therefore, taking all extensions into account, the total cost is
$O\left(d^2 \left(\sum_{i=1}^{n} v_i p_i^2 \log p_i\right)\right)$, 
which is expected to be sufficiently smaller than the FFT core complexity
$O\left(d^2 s \left(\sum_{i=1}^{n} v_i p_i\right)\right)$. 

Since the product $\zeta_{p_1^{v_1}}\zeta_{p_2^{v_2}}\cdots \zeta_{p_n^{v_n}}$ is a primitive $s$-th root of unity,  
the expansion of
$\bigl(X-\zeta_{p_1^{v_1}}\zeta_{p_2^{v_2}}\cdots \zeta_{p_n^{v_n}}\bigr)
\bigl(X-(\zeta_{p_1^{v_1}}\zeta_{p_2^{v_2}}\cdots \zeta_{p_n^{v_n}})^p\bigr)
\cdots
\bigl(X-(\zeta_{p_1^{v_1}}\zeta_{p_2^{v_2}}\cdots \zeta_{p_n^{v_n}})^{p^{d-1}}\bigr)
$
gives the minimal polynomial of $\zeta_s$.

When applying the Cantor-Zassenhaus algorithm to $\Phi_{p_i}$ or $X^{p_i} - \zeta_{{p_i}^{v-1}}$,  
we do not need to compute all irreducible factors.  
Instead, we apply it in the following simplified form. 
The algorithm below works only when $q$ is an odd prime power,  
but in the case where $q$ is a power of $2$, it can be modified using the technique described in~\cite{CZ}.
\begin{algorithm}[H]
\caption{The Cantor-Zassenhaus algorithm}
\begin{algorithmic}[1]
\Require{
$e \in \mathbb{Z}_{>0}$ and $f \in \mathbb{F}_q[X]$ such that $f$ is a product of $\frac{\deg f}{e}$ irreducible polynomials of degree $e$
}
\Ensure{An irreducible factor of $f$}
\State Randomly choose a monic polynomial $g \in \mathbb{F}_q[X]$ such that $0 < \deg g < \deg f$.
\If{$\gcd(g, f)$ is nontrivial}
  \State Replace $f$ with the one of smaller degree between $\gcd(g, f)$ and $f/\gcd(g, f)$.
\ElsIf{$\gcd(g^{\frac{q^e - 1}{2}} - 1, f)$ is nontrivial}
  \State Replace $f$ with the one of smaller degree between $\gcd(g^{\frac{q^e - 1}{2}} - 1, f)$ and $f/\gcd(g^{\frac{q^e - 1}{2}} - 1, f)$.
\EndIf
\If{$\deg f = e$}
  \State \Return $f$
\EndIf
\State Go to step 1.
\end{algorithmic}
\end{algorithm}

In order to execute the Cantor-Zassenhaus algorithm, it is necessary to determine the degree of each field extension.  
This can be done easily by using following propositions: 
\begin{prop}\label{Ext}
Let $p$ and $p_0$ be distinct prime numbers, and let $a$ be a positive integer relatively prime to both $p$ and $p_0$.  
If $p_0 \neq 2$, then the following two equalities hold for any integer $v > 1$:
\begin{align*}
[\mathbb{F}_p(\zeta_a, \zeta_{p_0}) : \mathbb{F}_p(\zeta_a)] &= \frac{\ord_{p_0} p}{\gcd([\mathbb{F}_p(\zeta_a) : \mathbb{F}_p], \ord_{p_0} p)}, \\
[\mathbb{F}_p(\zeta_a, \zeta_{p_0^v}) : \mathbb{F}_p(\zeta_a, \zeta_{p_0^{v-1}})] &=
\begin{cases}
1 & \text{if } v \leq v_{p_0}(p^{\ord_{p_0}p} - 1) + v_{p_0}([\mathbb{F}_p(\zeta_a) : \mathbb{F}_p]), \\
p_0 & \text{if } v > v_{p_0}(p^{\ord_{p_0}p} - 1) + v_{p_0}([\mathbb{F}_p(\zeta_a) : \mathbb{F}_p]).
\end{cases}
\end{align*}
\end{prop}

\begin{prop}\label{Ext2}
Let $p$ be an odd prime and let $a$ be an odd positive integer relatively prime to $p$.  
Then for any positive integer $v$, the following holds:
\[
[\mathbb{F}_p(\zeta_a, \zeta_{2^v}) : \mathbb{F}_p(\zeta_a, \zeta_{2^{v-1}})] =
\begin{cases}
1 & \text{if } v \leq v_2(p-1) + v_2(p+1) + v_2([\mathbb{F}_p(\zeta_a) : \mathbb{F}_p]) - 1, \\
2 & \text{if } v > v_2(p-1) + v_2(p+1) + v_2([\mathbb{F}_p(\zeta_a) : \mathbb{F}_p]) - 1.
\end{cases}
\]
\end{prop}

Based on Proposition~\ref{Ext} and Proposition~\ref{Ext2},  
we construct the following algorithm.  
Here, CZ stands for the Cantor-Zassenhaus algorithm.

\begin{algorithm}[H]
\caption{}
\begin{algorithmic}[1]
\Require{$p$ and $s=p_1^{v_1}p_2^{v_2}\cdots p_n^{v_n}$ ($v_i\geq 1$)}
\Ensure{An irreducible factor of $\Phi_s\in \mathbb{F}_p[X]$}
\State $F:=\mathbb{F}_p$
\State $a:=1$
\For{$i = 1, \ldots , n$}
  \If{$p_i\neq 2$}
      \State Introduce a fresh indeterminate $X_i$ and apply CZ to $\Phi_{p_i}(X_i)\in F[X_i]$ to compute $F(\zeta_{p_i})$
      \State Set $l:=v_{p_i}(p^{\ord_{p_i}p}-1)+v_{p_i}([F:\mathbb{F}_p])$
      \If{$v_i\leq l$}
        \State Apply CZ $v_i-1$ times to get $\zeta_{{p_i}^{v_i}}\in F(\zeta_{p_i})$
      \Else
        \State Apply CZ $l-1$ times to get $\zeta_{{p_i}^{l}}\in F(\zeta_{p_i})$
        \State Introduce a fresh indeterminate $X_i '$ and compute $F(\zeta_{{p_i}^{v_i}})$\\
        \Comment{$F(\zeta_{{p_i}^{v_i}})=F(\zeta_{{p_i}^l})[X_i ']/({X_i '}^{{p_i}^{v_i -l}}-\zeta_{{p_i}^l})$}
      \EndIf
  \Else
    \State Set $l:=v_2(p-1)+v_2(p+1)+v_2([F:\mathbb{F}_p])-1$
      \If{$v_i\leq l$}
        \State Apply CZ $v_i$ times to get $\zeta_{2^{v_i}}\in F$
      \Else
        \State Apply CZ $l$ times to get $\zeta_{2^{l}}\in F$
        \State Introduce a fresh indeterminate $X_i '$ and compute $F(\zeta_{2^{v_i}})$\\
        \Comment{$F(\zeta_{2^{v_i}})=F(\zeta_{2^l})[X_i ']/({X_i '}^{2^{v_i -l}}-\zeta_{2^l})$}
      \EndIf
  \EndIf
  \State Introduce a fresh indeterminate $Y_i$ and compute the minimal polynomial $f(Y_i)$ of $\zeta_{a{p_i}^{v_i}}$ over $\mathbb{F}_p$\\
  \Comment{$f(Y_i)=\bigl(Y_i-\zeta_a \zeta_{{p_i}^{v_i}}\bigr)\bigl(Y_i-(\zeta_a \zeta_{{p_i}^{v_i}})^p\bigr)\cdots \bigl(Y_i-(\zeta_a \zeta_{{p_i}^{v_i}})^{p^{[F(\zeta_{{p_i}^{v_i}}):F]-1}}\bigr)$}
  \State $F\gets \mathbb{F}_p(\zeta_{a{p_i}^{v_i}}):=\mathbb{F}_p[Y_i]/f(Y_i)$
  \State $a \gets a {p_i}^{v_i}$
\EndFor 

\State \Return $f(Y_n) \in \mathbb{F}_p[Y_n]$
\end{algorithmic}
\end{algorithm}

We now discuss the computational complexity.
The operation on line 24 requires $O(d^3 \log p)$ time to compute
$
\zeta_a \zeta_{{p_i}^{v_i}},\quad (\zeta_a \zeta_{{p_i}^{v_i}})^p,\quad \ldots,\quad (\zeta_a \zeta_{{p_i}^{v_i}})^{p^{[F(\zeta_{{p_i}^{v_i}}) : F] - 1}}.
$
The expansion of
$
\bigl(Y_i - \zeta_a \zeta_{{p_i}^{v_i}}\bigr) \bigl(Y_i - (\zeta_a \zeta_{{p_i}^{v_i}})^p\bigr) \cdots \bigl(Y_i - (\zeta_a \zeta_{{p_i}^{v_i}})^{p^{[F(\zeta_{{p_i}^{v_i}}) : F] - 1}}\bigr)
$
requires $O(d^3 \log d)$ time.
Therefore, the overall computational complexity is approximately
\[
O\left(d^2 \left(\sum_{i=1}^{n} v_i p_i^2 \log p_i\right) + n(d^3 \log p + d^3 \log d)\right).
\]


\subsection{Lifting the root of unity}
This subsection is the most important part of the paper.  
The results of the previous subsection—namely, the multiplicative FFT over finite fields—are unlikely to surpass other FFT methods on their own.  
In this subsection, we describe an algorithm for lifting the minimal polynomial obtained in the previous subsection. 

In this subsection, we do not rely on the assumption that $s$ is a product of small primes. Among the facts we use, the key one is simply that the polynomial $X^s - 1$ is sparse.
While our algorithm shares some ideas with Hubrechts' work on algorithms on unramified extensions of $p$-adic fields~\cite{Hubrechts}, 
a notable difference is our focus on exploiting the sparsity of the polynomial to reduce computational complexity. 

When computing roots of unity over the complex numbers, one can use the identity  
$\zeta_s = \cos\left(\frac{2\pi}{s}\right) + \sqrt{-1} \sin\left(\frac{2\pi}{s}\right)$,  
along with the Maclaurin expansions of $\cos$ and $\sin$ to compute them numerically.  
The situation here is therefore quite different from the complex case.

The lifting procedure is essentially based on Hensel lifting,  
but we present an algorithm that performs significantly faster than executing Lemma~\ref{HLift}.  
We first give a detailed explanation of the general form of Hensel's lemma.
There are two commonly referenced versions of Hensel's lemma.  
One is Corollary~\ref{HLGeneral}, and the other is the following statement.
Mathematically, the proposition below can be interpreted as a special case of Corollary~\ref{HLGeneral}  
when $\deg f = 1$, that is, when $f = X - \alpha$ and $g = h / (X - \alpha)$.

\begin{prop}\label{HLiftFast}
Let $A$ be a complete DVR, and let $\mathfrak{m}$ be its maximal ideal.  
Let $h \in A[X]$ be a monic polynomial, and let $\alpha \in A$ satisfy  
$h(\alpha) \in \mathfrak{m}$ and $h'(\alpha) \notin \mathfrak{m}$.  
Then there exists a unique $\alpha' \in A$ such that  
$h(\alpha') = 0$ and $\alpha' - \alpha \in \mathfrak{m}$.
\end{prop}

These two versions of Hensel's lemma have distinct standard proofs.  
In particular, this distinction leads to a decisive difference in computational complexity.  
Proposition~\ref{HLiftFast} can be proved as follows.

\begin{proof}
Let $k$ be a positive integer, and suppose that some $\beta \in A$ satisfies $h(\beta) \in \mathfrak{m}^k$ and $h'(\beta) \notin \mathfrak{m}$.  
Since $h$ is a polynomial, we can expand it as 
\begin{equation*}
h(\beta + X) = h(\beta) + \frac{h'(\beta)}{1!} X + \frac{h''(\beta)}{2!} X^2 + \cdots.
\end{equation*}  
By substituting $X = -\frac{h(\beta)}{h'(\beta)}$, we obtain  
$h\left(\beta - \frac{h(\beta)}{h'(\beta)}\right) \in \mathfrak{m}^{2k}$  
and  
$h'\left(\beta - \frac{h(\beta)}{h'(\beta)}\right) \notin \mathfrak{m}$. 
Therefore, starting with $\beta = \alpha$, we can iteratively update $\beta$ at the $k$-th step by replacing it with  
$\beta - \frac{h(\beta)}{h'(\beta)} \mod \mathfrak{m}^{2^k + 1}$  
to lift the root.
Uniqueness follows by the same argument.
\end{proof}

This proof is essentially Newton's method, and it lifts the root with quadratic convergence.  
By contrast, the lifting procedure in Lemma~\ref{HLift} has two disadvantages:
\begin{itemize}
  \item Only one digit (or one level of precision) is obtained per step, and
  \item If multiplication and reduction are performed without FFT, each step incurs a computational cost of about $\Omega(\deg^2 f + \deg^2 g)$ arithmetic operations over $A$.
\end{itemize}
In this respect, Proposition~\ref{HLiftFast} is superior.

\begin{rem}
The lifting method in Lemma~\ref{HLift} is referred to as \emph{linear lifting}.  
There exists a method called \emph{quadratic lifting}, which achieves higher precision in a single step.  
However, it requires an additional cost of approximately $\Omega(\deg^2 f + \deg^2 g)$ per step,  
which is too expensive for our purposes.
\end{rem}

This very fast lifting method is restricted to the case where $\deg f = 1$,  
but we extend it to cases where $\deg f \ll \deg h$ to enable efficient computation of roots of unity.

The rough idea is as follows. Let $h = \Phi_s$.  
Instead of computing some irreducible factor $\tilde{f}$ such that $\Phi_s = \tilde{f} \tilde{g}$ over the DVR $\mathbb{Z}_p$  
(by using Corollary~\ref{HLGeneral}),  
we select a polynomial $F \in \mathbb{Z}_p[X]$ that agrees with $\tilde{f}$ modulo $p$,  
and apply Proposition~\ref{HLiftFast} to the root $X \in \mathbb{Z}_p/F$ of $\Phi_s \bmod p$  
over the DVR $\mathbb{Z}_p[X]/F$.  
That this ring is a DVR follows from Proposition~\ref{Unramified}.

In other words, we apply Proposition~\ref{HLiftFast} with  
$A = \mathbb{Z}_p/F$, $\mathfrak{m} = pA$, $h = \Phi_s$, and $\alpha = X$.

Of course, there is some freedom in the choice of $F$.  
For example, $\Phi_5$ factors as $(X^2 + 5X + 1)(X^2 + 15X + 1)$ over $\mathbb{F}_{19}$.  
If we take $f = X^2 + 5X + 1$, then possible choices for $F \in \mathbb{Z}_p[X]$ include:
\begin{itemize}
  \item $X^2 + 5X + 1$,
  \item $X^2 - 14X - 18$, and
  \item $X^2 + (5 + 1 \cdot 19 + 1 \cdot 19^2 + 1 \cdot 19^3 + \cdots) X + 1$.
\end{itemize}
Any of these choices are valid.  
However, from the perspective of computational cost, it is better to choose an $F$ such that reduction modulo $F$ is computationally efficient. 

To formalize this idea, we rely on the following result.

\begin{thm}
Let $h \in \mathbb{Z}_p[X]$ be a monic polynomial, and let $f \in \mathbb{F}_p[X]$ be a monic irreducible polynomial of degree at least one.  
Let $\bar{h}$ denote the image of $h$ in $\mathbb{F}_p[X]$, and suppose that $f$ divides $\bar{h}$ but $f^2$ does not.  
(In this case, by Corollary~\ref{HLGeneral}, $f$ lifts to a monic factor of $h$ over $\mathbb{Z}_p$.)

Let $F \in \mathbb{Z}_p[X]$ be any polynomial whose image in $\mathbb{F}_p[X]$ equals $f$. Then the following statements hold:
\begin{itemize}
  \item $A := \mathbb{Z}_p / F$ is a complete DVR with maximal ideal $pA$.
  \item Let $\bar{X}$ denote the image of $X$ in $A/pA$. Then the expansion of  
  $(Y - \bar{X})(Y - \bar{X}^p) \cdots (Y - \bar{X}^{p^{\deg f - 1}})$  
  in $(A/pA)[Y]$ coincides with $f(Y)$.
  \item For each $i = 0, 1, \ldots, \deg f - 1$, we have $\bar{h}'(\bar{X}^{p^i}) \neq 0$.
  \item Let $X_i'$ denote the result of applying Proposition~\ref{HLiftFast} with $A = \mathbb{Z}_p / F$, $h = h$, and $\alpha = X^{p^i}$ (for $i = 0, 1, \ldots, \deg f - 1$).  
  Then the expansion of  
  $(Y - X_0')(Y - X_1') \cdots (Y - X_{\deg f - 1}')$  
  lies in $\mathbb{Z}_p[Y]$ and equals $\tilde{f}(Y)$,  
  the lifted polynomial obtained via Corollary~\ref{HLGeneral} with $A = \mathbb{Z}_p$, $h = h(Y)$, $f = f(Y)$, and $g = \bar{h}(Y)/f(Y)$.
\end{itemize}
\end{thm}
\begin{proof}
That $A$ is a DVR with maximal ideal $pA$ follows from Proposition~\ref{Unramified}.  
That $A$ is complete follows immediately from the completeness of $\mathbb{Z}_p$.

Since $\bar{X} \in A/pA$ is a root of $f(Y) \in \mathbb{F}_p[Y] \subset A/pA[Y]$,  
the elements $\bar{X}^p, \bar{X}^{p^2}, \ldots, \bar{X}^{p^{\deg f - 1}}$ are also roots.  
These elements $\bar{X}, \bar{X}^p, \ldots, \bar{X}^{p^{\deg f - 1}}$ are all distinct.  
As the degree of $f$ equals that of the product  
$(Y - \bar{X})(Y - \bar{X}^p) \cdots (Y - \bar{X}^{p^{\deg f - 1}})$,  
we conclude that  
$f(Y) = (Y - \bar{X})(Y - \bar{X}^p) \cdots (Y - \bar{X}^{p^{\deg f - 1}})$.

Suppose for contradiction that there exists some $i$ such that $\bar{h}'(\bar{X}^{p^i}) = 0$.  
Then $\bar{h}$ has a multiple root at $\bar{X}^{p^i}$.  
Since $\bar{h} = f \cdot \frac{\bar{h}}{f}$ and $f$ has a simple root at $\bar{X}^{p^i}$,  
it must be that $\bar{h}/f$ also vanishes at $\bar{X}^{p^i}$.  
As $\bar{h}/f$ has coefficients in $\mathbb{F}_p$, this implies that $f$ divides $\bar{h}/f$,  
which contradicts the assumption that $\bar{h}$ is not divisible by $f^2$.

By applying Corollary~\ref{HLGeneral} with $A = A$, $h = h(Y)$, $f = f(Y)$, and $g = \bar{h}(Y)/f(Y)$,  
the lifted polynomial of $f$ must coincide with $\tilde{f}$ by uniqueness of lifting.  
Furthermore, since $(Y - X_0')(Y - X_1') \cdots (Y - X_{\deg f - 1}')$ agrees with $\tilde{f}$ modulo $pA[Y]$,  
we conclude that $\tilde{f} = (Y - X_0')(Y - X_1') \cdots (Y - X_{\deg f - 1}')$.
\end{proof}

The above discussion applies not only to $\Phi_s$, but to the lifting of roots of arbitrary polynomials.  
However, when we restrict our attention to the lifting of roots of $\Phi_s$, the following observation leads to further acceleration:

\begin{quote}
Instead of lifting roots of $\Phi_s$, we may lift roots of $X^s - 1$.  
Since $X^s - 1$ is a sparse polynomial, this results in a dramatic reduction in computational cost.  
In other words, given an approximation $\alpha$ to a root of $\Phi_s$,  
one Hensel lifting step can be carried out simply by computing $\alpha^s - 1$ and $1 / (s \alpha^{s - 1})$ via repeated squaring.
\end{quote}

Before presenting the algorithm, we state the following lemma, which is closely related to the above remark.

\begin{lem}
Let $A$ be a DVR with maximal ideal $\mathfrak{m}$, and let $\alpha \in A$, $k \in \mathbb{Z}_{>0}$, and $s \in \mathbb{Z}_{>0}$ such that $\alpha^s \equiv 1 \pmod{\mathfrak{m}^k}$.  
Then the following holds:
\[
\frac{1}{\alpha^{s - 1}} \equiv \alpha (2 - \alpha^s) \pmod{\mathfrak{m}^{2k}}.
\]
\end{lem}

\begin{proof}
Since $\alpha \alpha^{s - 1} \equiv 1 \pmod{\mathfrak{m}^k}$,  
computing the inverse of $\alpha^{s - 1}$ to precision $\mathfrak{m}^{2k}$  
requires applying one step of Newton's method starting from $\alpha$.  
This completes the proof.
\end{proof}

The algorithm is given below in pseudocode.


\begin{algorithm}[H]
\caption{Lifting an irreducible factor of a cyclotomic polynomial}
\label{LiftAlgo}
\begin{algorithmic}[1]
\Require{
$n \in \mathbb{Z}_{>0}$ and an irreducible factor $\bar{f}$ of a cyclotomic polynomial over $\mathbb{F}_p$
}
\Ensure{
A lift of $\bar{f}$ to $\mathbb{Q}_p$ modulo $p^{2^n}$
}
\State Find and fix $F \in \mathbb{Z}_p[X]$ such that $F \bmod p = \bar{f}$.
\State Set $A := \mathbb{Z}_p[X]/F$.
\State Let $\alpha$ be the image of $X$ in $A$.
\For{$i = 1, 2, \ldots, n$}
  \State $\alpha \gets \alpha - \frac{1}{s} (\alpha^s - 1) \alpha (2 - \alpha^s) \bmod p^{2^i}$\\
  \Comment{At the start of the $i$-th loop, we have $\frac{1}{\alpha^{s-1}} \equiv \alpha \pmod{p^{2^{i-1}}}$}
\EndFor
\State Introduce a fresh indeterminate $Y$.
\State \Return $(Y - \alpha)(Y - \alpha^p) \cdots (Y - \alpha^{p^{\deg f - 1}}) \bmod p^{2^n} \in A[Y]$
\end{algorithmic}
\end{algorithm}

We now analyze the computational complexity of this algorithm.  
First, the loop beginning on line 4 incurs a cost of $O(d^2 \log s)$ per iteration,  
resulting in a total complexity of $O(d^2 n \log s)$. 
The computation in line 7 involves calculating $\alpha^p, \alpha^{p^2}, \ldots, \alpha^{p^{\deg f - 1}}$, 
which takes approximately $O(d^3 \log p)$ time, 
and expanding the product  
$(Y - \alpha)(Y - \alpha^p) \cdots (Y - \alpha^{p^{\deg f - 1}}) \bmod p^{2^n}$ 
takes about $O(d^3 \log d)$ operations.

\begin{rem}
Since our goal is to perform FFT efficiently,  
we may carry out FFT computations directly in $\mathbb{Q}_p / F$ rather than in $\mathbb{Q}_p[X]/f$,  
where $f$ is the minimal polynomial of $\zeta_s$.  
Doing so also allows us to omit line 7 of Algorithm~\ref{LiftAlgo}.
\end{rem}

\begin{rem}
The method for computing roots of unity by applying Newton's method on $X^s-1$ is also applicable over the complex numbers.  
This fact is rarely emphasized in explanations of FFT,  
possibly because the classical method using the Maclaurin expansions of $\sin$ and $\cos$  
already provides sufficiently fast and practical performance.
\end{rem}

\section{Fast Fourier transform for polynomials over $\mathbb{Q}_p$}
Using the algorithm described in the previous section, we construct an asymptotically fast FFT over $\mathbb{Q}_p$ for a fixed prime $p$.  
By an FFT over $\mathbb{Q}_p$, we mean an algorithm that, given as input a polynomial $f \in \mathbb{Q}_p[X]$ and a positive integer $N$,  
selects an integer $s > N$ and outputs the vector $(f(1), f(\zeta_s), \ldots, f(\zeta_s^{s-1}))$.  
Our algorithm runs in time 
$O(N^{1+o(1)}) \biggl(:= \bigcap\limits_{\varepsilon > 0} O(N^{1+\varepsilon})\biggr)$. 

Given an input pair $(f, N)$, let $r$ be the smallest positive integer such that
$N < \Phi_1(p)\Phi_{p_1}(p)\Phi_{p_2}(p)\cdots\Phi_{p_r}(p)$, 
and set
$s = \Phi_1(p)\Phi_{p_1}(p)\Phi_{p_2}(p)\cdots\Phi_{p_r}(p)$. 
Then, by elementary arguments, we have the extension degree
$d := [\mathbb{Q}_p(\zeta_s) : \mathbb{Q}_p] = p_1 p_2 \cdots p_r$.
We compute $\zeta_s$ using the methods described in the previous section, and then perform the FFT.  
To factor $s$, it suffices to factor each of $\Phi_1(p), \Phi_{p_1}(p), \Phi_{p_2}(p), \ldots, \Phi_{p_r}(p)$.

We now evaluate the complexity of this FFT algorithm.  
The core FFT computation has cost
$O\bigl(d^2 s(\Phi_1(p) + \Phi_{p_1}(p) + \Phi_{p_2}(p) + \cdots + \Phi_{p_r}(p))\bigr)$, 
so it remains to show that this is $O(N^{1+o(1)})$.
By the prime number theorem, $p_i \approx i \log i$.  
Therefore, the ratio
$\frac{\Phi_{p_r}(p)}{\Phi_1(p)\Phi_{p_1}(p)\Phi_{p_2}(p)\cdots\Phi_{p_r}(p)}$
tends to zero as $N \to \infty$, implying that $\frac{s}{N} \to 1$ as $N \to \infty$.  
Furthermore, for any $\varepsilon > 0$, we can show that
$\frac{\Phi_1(p)+\Phi_{p_1}(p)+\Phi_{p_2}(p)+\cdots + \Phi_{p_r}(p)}{N^{\varepsilon}} \to 0 \quad \text{as } N \to \infty$.
Since $p_i \approx i \log i$ and the choice of $r$ satisfies $\log N \approx r^2 \log r$,  
we deduce that $r \approx \sqrt{\log N}$.  
From this, we have
$d < p_r^r \approx (r \log r)^r = \exp\left(r \log(r \log r)\right) \approx \exp\left(\sqrt{\log N} \log(\sqrt{\log N} \log \sqrt{\log N})\right)$, 
and thus
\[
\frac{d^2}{N^{\varepsilon}} \approx \exp\left(2 \sqrt{\log N} \log(\sqrt{\log N} \log \sqrt{\log N}) - \varepsilon \log N\right),
\]
which tends to zero as $N \to \infty$.
Therefore, the total complexity
$d^2 s(\Phi_1(p) + \Phi_{p_1}(p) + \Phi_{p_2}(p) + \cdots + \Phi_{p_r}(p))
\in O(N^{1+o(1)})$
as claimed.


\end{document}